\pdfoutput=1
\documentclass[a4paper,USenglish]{lipics-v2021}
\usepackage[utf8]{inputenc}
 
\usepackage{algorithmic}
\usepackage{textcomp}
\usepackage{xcolor}
\def\BibTeX{{\rm B\kern-.05em{\sc i\kern-.025em b}\kern-.08em
    T\kern-.1667em\lower.7ex\hbox{E}\kern-.125emX}}

\usepackage{comment}
\usepackage{amsthm}
\usepackage{amsmath}

\usepackage{color}
\nolinenumbers
\hideLIPIcs

\newcommand{\kaushik}[1]{{\color{red}\underline{\textsf{Kaushik:}}} {\color{red} \emph{#1}}}

\title{Fast Deterministic Gathering with Detection on Arbitrary Graphs: The Power of Many Robots}

\author{Anisur Rahaman Molla} {Indian Statistical Institute, Kolkata, India}{molla@isical.ac.in}{https://orcid.org/0000-0002-1537-3462}{A.\,R. Molla was supported, in part, by DST INSPIRE Faculty Research Grant DST/INSPIRE/04/2015/002801, Govt. of India and ISI DCSW/TAC Project, file number E5412.}

\author{Kaushik Mondal} {Indian Institute of Technology Ropar, Ropar, India}{kaushik.mondal@iitrpr.ac.in}{https://orcid.org/0000-0002-9606-9293}{K. Mondal was partially supported by the FIST program of the Department of Science and Technology, Government of India, Reference No. SR/FST/MS-I/2018/22(C).}

\author{William K. Moses Jr.\footnote{Part of the work was done while William K. Moses Jr. was a post doctoral fellow at the University of Houston in Houston, USA.}} {Durham University, Durham, UK}{william.k.moses-jr@durham.ac.uk}{https://orcid.org/0000-0002-4533-7593}{W.\,K. Moses Jr.\ was supported in part by NSF grants CCF-1540512, IIS-1633720, and CCF-1717075 and in part by BSF grant 2016419.}

\authorrunning{A.\,R. Molla, K. Mondal, and W.\,K. Moses Jr.} 

\Copyright{Anisur Rahaman Molla and Kaushik Mondal and William K. Moses Jr.}



\begin{document}

\maketitle

\begin{abstract}
Over the years, much research involving mobile computational entities has been performed. From modeling actual microscopic (and smaller) robots, to modeling software processes on a network, many important problems have been studied in this context. Gathering is one such fundamental problem in this area. The problem of gathering $k$ robots, initially arbitrarily placed on the nodes of an $n$-node graph, asks that these robots coordinate and communicate in a local manner, as opposed to global, to move around the graph, find each other, and settle down on a single node as fast as possible. A more difficult problem to solve is gathering \textit{with detection}, where once the robots gather, they must subsequently realize that gathering has occurred and then terminate.

In this paper, we propose a \textit{deterministic} approach to solve gathering with detection for any \textit{arbitrary connected graph} that is \textit{faster} than existing deterministic solutions for even just gathering (without the requirement of detection) for arbitrary graphs. In contrast to earlier work on gathering, it leverages the fact that there are more robots present in the system to achieve gathering with detection faster than those previous papers that focused on just gathering. The state of the art solution for deterministic gathering~[Ta-Shma and Zwick, TALG, 2014] takes $\Tilde{O}(n^5 \log \ell)$ rounds, where $\ell$ is the smallest label among robots and $\Tilde{O}$ hides a polylog factor. 
We design a deterministic algorithm for gathering with detection with the following trade-offs depending on how many robots are present: (i) when $k \geq \lfloor n/2 \rfloor + 1$, the algorithm takes $O(n^3)$ rounds, (ii) when $k \geq \lfloor n/3 \rfloor + 1$, the algorithm takes $O(n^4 \log n)$ rounds, and (iii) otherwise, the algorithm takes $\Tilde{O}(n^5)$ rounds. The algorithm is not required to know $k$, but only $n$.

\end{abstract}

\section{Introduction}
\label{sec:intro}
A fundamental area of interest in distributed computing relates to when computational entities are allowed to move around in some fixed space and interact with each other. Research of this type has broad implications ranging from designing matter that can be programmed to respond to external stimuli~\cite{GCM05} to swarm robotics~\cite{STZLW12} to even modern self-driving car technology~\cite{BGCACFJBPMV21}; when considering tasks over a network, modeling processes via mobile robots can result in benefits for a host of applications such as distributed information retrieval, e-commerce, information dissemination, and workflow applications and groupware~\cite{LO99}. When that space is discretized, the popular framework of mobile robots on graphs is used to study important problems in this area. This framework is especially useful to study problems related to real world robots that must move in structures with rooms and corridors. It is also useful to study the behavior of software agents that may travel in the internet from computer to computer. 

The problems that are actually studied in this setting usually take on the form of either having the robots work together to find something in the graph (e.g., exploration~\cite{Bampas:2009,Cohen:2008,Das13,Dereniowski:2015,Fraigniaud:2005,MencPU17}, treasure hunting~\cite{MP15})
or form a certain configuration (e.g., gathering~\cite{CFPS12,CP02,DKLHPW11,P07}, dispersion~\cite{Augustine:2018,KMS2019,MMM21},  scattering~\cite{Barriere2009,ElorB11,Poudel18,Shibata:2016}, pattern formation \cite{SY99},
convergence \cite{CP04}). In this paper, we look at the problem of gathering with detection (a more difficult to solve variant of gathering), where multiple robots, initially arbitrarily placed on a graph, must find each other on the same node, become aware that gathering is complete and subsequently terminate.\footnote{In the past, gathering has also been referred to as the rendezvous problem when there are just two robots; for consistency, we use the term gathering throughout the paper.}
While this is an abstract problem, solutions to it may be adapted to the more concrete situation where you have multiple humans or robots trying to find each other in a discretized space such as in a maze with rooms and corridors between them or in cities with roads and intersections. Furthermore, gathering with detection can also be used as a subroutine when solving problems like exploration, scattering, dispersion, etc.\ as solving those problems when all robots are gathered is easier than solving them when robots are arbitrarily placed in the graph.

When trying to solve this problem, a fundamental difference when designing solutions arises from whether the robots have access to a source of randomness or not. While randomness may allow for faster solutions, we want to focus on the setting where robots do not have access to such a source of additional power. In this setting, we present an algorithm to solve gathering that is faster than the current state of the art with respect to time, assuming that all robots are initially awake. A key insight used in our algorithm is that the number of robots present in the graph may be considered an additional power to the algorithm designer and may be leveraged. Some example tasks where many robots might be used include exploring an area~\cite{Dereniowski:2015,FraigniaudGKP06,OS14}, dispersing the robots over a given area~\cite{Augustine:2018,KMS2019,MM22}, searching for treasure, etc. After the task is completed, we may want to gather the robots.

This is in contrast to previous work in the setting we consider, where a solution to the problem of gathering when just two robots are present is provided and it is assumed that it can be generalized to multiple robots; this approach is usually adopted because it has been shown that solutions to gathering of two robots can be easily extended to handle multiple robots~\cite{KM08}.
We believe that providing an example of such an approach to solving the problem will pave the way for more such faster solutions. To highlight this, notice that the problem of gathering was introduced more than 60 years ago~\cite{S60} and the previous best known deterministic solution~\cite{TZ07,TZ14}, which itself is more than 14 years old, takes $\Tilde{O}(n^5 \log \ell)$ time, where $\ell$ is the smallest label among all robots. Our solution for the more difficult problem of gathering with detection, when the number of robots is sufficiently large, can take as little as $O(n^3)$ time.

\subsection{Model}\label{subsec:model}

Consider an arbitrary connected undirected graph with $n$ nodes and $m$ edges where the nodes are unlabeled. Each node has port numbers corresponding to each of its edges. The ports of a node have unique labels in $[0,\delta-1]$ where $\delta$ is the degree of the node. Note that an edge between adjacent nodes may have different port numbers assigned to it by the corresponding nodes. 

There are a total of $k$ robots in the system.\footnote{In the literature, the computational entities have also been referred to as ``mobile agents''. We use ``mobile robots'' throughout for consistency.} Each robot has a unique identifier (ID), also called its label, assigned to it from the range $[1,n^b]$, where $b>1$ is a constant. Note that robots need not have the same length bit string for their IDs. Two robots co-located on the same node can communicate with each other via messages to each other. Typically, for the gathering problem, an assumption is made that when robots meet, they can become aware of this fact and this is considered sufficient. However, as noted in~\cite{EP17}, detection is intimately tied to how robots communicate with one another. In~\cite{EP17}, they considered that robots communicate via the beeping model~\cite{CK10}, which can be argued to be one of the weakest communication models and in~\cite{BDP20}, they adopt an even weaker notion that robots can only detect the number of other robots co-located with it at a node. We opt for the message passing framework where multiple robots, co-located on the same node in the same round, can exchange messages with each other in that round. This framework is also known as the Face-to-Face model and has been used to study other problems related to mobile robots including exploration~\cite{D19} and dispersion~\cite{AM18,MM22}. If a robot moves from one node to an adjacent node, it is aware of both port numbers assigned to the edge through which it passed. We assume that each robot knows the value of $n$ as well as its own unique ID. We highlight the fact that robots do not know the values of $k$, $b$, or any graph parameters other than $n$. 

We consider a synchronous system that proceeds in rounds. In each round, the robots perform the following sequence of operations: (i) robots co-located at the same node may communicate with each other via message passing; robots perform local computation as required (ii) each robot, if it decided to in step (i), performs a movement along an edge to an adjacent node. 
In the framework of mobile robots on a graph, the focus is on bounding the amount of movement required by the robots as this is considered to be far more expensive, time-wise, than local computation and communication. As such, a common assumption for problems studied in this setting (e.g., exploration, dispersion) is that local computation and communication is free and any amount of communication and computation may be performed in step 
(i) of a round. While this assumption is beneficial, to ensure the practicality of our algorithms, we design them so that any local computation performed is at most polynomial in $n$. 
We make the assumption that all robots are initially awake and start any prescribed algorithm at the same time. 

We now formally state the problems of gathering and gathering with detection.

\medskip
\noindent \textbf{Gathering:} Initially, $k$ robots are placed on the nodes of an $n$-node graph. Design an algorithm to be run by each robot such that eventually all $k$ robots meet at one node.

\medskip
\noindent \textbf{Gathering With Detection:} Initially, $k$ robots are placed on the nodes of an $n$-node graph. Design an algorithm to be run by each robot such that eventually all $k$ robots meet at one node and subsequently terminate.

\subsection{Our Contributions}\label{subsec:our-contrib}

Our main contribution is to showcase the first approach, to the best of our knowledge, of leveraging many robots to achieve fast deterministic gathering with detection of mobile robots in graphs. To be clear,  multiple robots can be used in gathering for other purposes such as to break symmetry when robots have no labels, but to the best of our knowledge they have not been leveraged to achieve a faster running time. Also, we are specifically talking about the problem of gathering of mobile robots on a graph. The problem of exploration of mobile robots on a graph has already seen this approach being used, resulting in very fast algorithms~\cite{Dereniowski:2015}.

In this paper, we propose a \textit{deterministic} approach to solve gathering with detection for any \textit{arbitrary graph} that is \textit{faster} than existing deterministic solutions for even just gathering (without the requirement of detection) for arbitrary graphs. In contrast to earlier work on gathering, it leverages the fact that there are more robots present in the system to achieve gathering with detection faster than those previous papers that focused on just gathering. The state of the art solution for deterministic gathering on arbitrary graphs~\cite{TZ14} takes $\Tilde{O}(n^5 \log \ell)$, where $\ell$ is the smallest label among robots and $\Tilde{O}$ hides polylog factors of $n$.  
We design a deterministic algorithm for gathering with detection with the following trade-offs depending on how many robots are present: (i) when $k \geq \lfloor n/2 \rfloor + 1$, the algorithm takes $O(n^3)$ rounds, (ii) when $k \geq \lfloor n/3 \rfloor + 1$, the algorithm takes $O(n^4 \log n)$ rounds, and (iii) otherwise (for any $k < \lfloor n/3 \rfloor + 1$), the algorithm takes $\Tilde{O}(n^5)$ rounds.
Each robot only requires $O(m \log n)$ bits of memory.\footnote{Note that we do \emph{not} require robots to know the value of $m$, we just require that they be provided sufficient memory.}

It is clear to see that in the situation where a sufficient number of robots are present, i.e., $k \geq \lfloor n/3 \rfloor + 1$, our algorithm solves gathering with detection asymptotically faster than any previous algorithm for even gathering. 
Furthermore, it should be noted that the purpose of requiring many robots is to ensure that there exist robots that are sufficiently close together in the initial configuration. As such, regardless of the number of robots available, if there exist two robots that are at a hop distance of $i$ apart, then the following holds. If $i = 0,1,2$, then the algorithm achieves gathering with detection in $O(n^3)$ rounds. If $i= 3, 4$, then the algorithm achieves gathering with detection in $O(n^4 \log n)$ rounds. For $i \geq 5$, the algorithm achieves gathering with detection in $\tilde{O}(n^5)$ rounds. 

\subsection{Technical Difficulties and High-Level Ideas}\label{subsec:tech-diff}
The hardest part of gathering is to ensure that the robots have an effective exploration strategy of the graph that results in robots eventually meeting. Normally, this is done by providing each robot with a sequence of moves to perform and a schedule to perform them such that eventually the whole graph is explored by each robot and there exists a sequence of rounds such that for every pair of robots, one of them will be stationary while the other one moves. The state of the art exploration strategy has been to use a Universal Exploration Sequence (UXS) and couple that with stationary rounds and moving rounds corresponding to bits in the ID string of a robot. While this approach works, the time it takes to use a UXS to explore the graph is high.

One of the key ideas we leverage is that when more robots are present on the graph, we can infer something about the distance between the two closest robots in the starting configuration. More specifically, we show that when $k \geq \lfloor n/c \rfloor + 1$, there exist two robots which are at most $2c-2$ hop distance from each other. 

This insight is crucial, but not sufficient by itself to achieve fast gathering. In fact, it is mentioned in Dessmark et al.~\cite{DFKP06} that when two robots start simultaneously and the distance between them is bounded by $D$, gathering can be achieved in $O(D \Delta^D \log \ell)$ rounds, where $\Delta$ is the maximum degree of the graph. However, their result was for two robots, and if we think about extending it to many robots, it appears that $D$ would reflect the maximum shortest hop distance between any two robots in the initial configuration. Since robots are thought to be placed by an adversary, this running time can be quite large, even when there are plenty of robots on the graph.\footnote{A possible adversarial placement is one that maximizes the running time, for example by keeping the robots as far away from each other as possible. For example, suppose all robots are in two groups, and these groups are $D-1$ distance from each other. If $\Delta = \Omega(n)$ and $D = \Omega(n)$, then the running time of Dessmark et al.'s algorithm would be exponential. For randomized placement, we may not get such bad configurations in every placement with high probability.}

Thus, we must find some way to cleverly leverage the fact that there exist two robots close together. First of all, we have these robots meet one another in a way similar to that of Dessmark et al.~\cite{DFKP06}, i.e., we have them perform a neighbor search procedure to find one another. Once these two robots meet, we then utilize them to gather the remaining robots.\footnote{Note that it may be the case that two \textit{or more} robots end up meeting. This is functionally the same as just having two robots meet so we ignore this situation when describing the high level idea; we explicitly handle this situation in the algorithm description.} We do this by first having these two robots construct a map of the graph using a simple token-explorer style algorithm where one robot acts as a movable token while the other robot grows its knowledge of the graph in balls of increasing radius centered at the node they first met at. Subsequently, these two robots go around collecting the other robots.

However, we must account for situations where there might be multiple such pairs of robots. We want to ensure that eventually, all robots are gathered together. When some robots move and some remain stationary, we must design a strategy to ensure all robots are ``captured'' and gather together. To this end, we have each pair of robots from the map construction step do the following. One stays put (to capture other moving robots) and the other robot moves around the whole graph in a systematic way (to capture other stationary robots). The exact nature of who captures who depends on the smallest label among each pair of robots, which acts as a sort of ID for the pair. If a robot from a lower ID pair comes into contact with a robot from a higher ID pair, the higher ID robot is captured by the lower ID robot, i.e., the higher ID robot stops whatever it was doing and ``follows'' the lower ID robot henceforth. In this manner, the pair with the lowest ID will capture all other robots and subsequently gather them together.

Finally, a word must be said about the catch-all case, i.e., when there are not enough robots or when there do not exist robots that are sufficiently close together to allow for faster gathering with detection as described above. In this situation, we default to using a Universal Exploration Sequence (UXS) to help robots find one another, similar to the idea in Ta-Shma and Zwick~\cite{TZ14}. However, we introduce our own twist to ensure that gathering \textit{with detection} occurs. Since the value of $n$, the number of nodes in the graph, is known to all robots, they can deterministically compute the same UXS that guarantees that they can explore the graph, say of length $T$. Now, we have robots perform the following sequence of actions in phases of $2T$ rounds. For each robot, it reads its label from least significant bit to the most significant bit. If the bit is $0$, then the robot waits at its current node for $T$ rounds and then explores the graph using the UXS for the next $T$ rounds. If the bit is $1$, the robot does things in the opposite order. If a robot comes across another robot with a higher ID, it starts to follow that higher ID robot, i.e., it implements choices according to the ID bits of the higher ID robot. If a robot runs out of bits, then it stays at its node for $2T$ rounds. In this time, if no robot with a higher ID (or higher ID than the robot it is following) shows up, the robot  decides that gathering has been achieved and subsequently terminates.

\subsection{Related Work}\label{subsec:rel-work}

Gathering is a very old problem, first studied more than 60 years ago in Schelling~\cite{S60}. Subsequent research on the problem has been divided broadly into whether the underlying medium of movement is a continuous space (e.g., Euclidean plane) or a discretized space (e.g., graphs). We focus on graphs; Flocchini~\cite{F19} is a good survey of work on gathering in a plane. A further divide occurs when considering whether the solutions require robots to have access to randomness or not. In the current work, we focus on deterministic solutions to gathering on a graph; Alpern~\cite{A02} is a good survey of randomized solutions (not just on a graph). Our focus in this paper is on the specific scenario where the robots have unique labels; Pelc~\cite{P19} is a good survey with a section devoted to when robots are anonymous. 

For the currently considered setting of deterministic gathering of labeled robots, there has been a bit of work on arbitrary graphs~\cite{DFKP06,KM08,TZ14}; there has also been work on gathering with detection of two robots on arbitrary graphs~\cite{EP17} and gathering with detection of several robots on arbitrary graphs~\cite{BDP20}. 
Dessmark et al.~\cite{DFKP06} studied gathering of two robots in a variety of settings including on arbitrary graphs with a startup delay of $\tau$ between the starting times of the two robots. In such a setting, time is measured from the moment the final robot wakes up. They presented the first deterministic algorithm to solve gathering in an arbitrary connected graph in time polynomial in $\tau$, $n$, and $\ell$, where $\ell$ is the smaller of the two robots' labels. Specifically, the running time of their algorithm was $O(n^5 \sqrt{\tau \log \ell } \log n + n^{10} \log^2 n \log \ell)$ rounds. They also showed that when it is assumed that robots start simultaneously, i.e., $\tau = 0$, gathering can be accomplished in $O(D \Delta^D \log \ell)$ rounds, where $D$ is the initial distance between the two robots and $\Delta$ is the maximum degree of the graph. Since robots are thought to be placed on the graph by an adversary, and there exist graphs with diameter of the order of $\Omega(n)$, this algorithm takes an exponential time in $n$ on some graphs.
Kowalski and Malinowski~\cite{KM08} show how to deterministically achieve gathering, even when there is a non-zero startup delay between robots, in $O(\log^3 \ell  + n^{15} \log^{12} n)$ rounds. In addition, they show how any solution for gathering of two robots can be adapted to solve gathering of an arbitrary number of robots ($>2$) in the same running time. 
The current state of the art, Ta-Shma and Zwick~\cite{TZ14} reduced the running time of gathering, again even if there is a non-zero startup delay, to $\Tilde{O}(n^5 \log \ell)$ rounds, where the $\Tilde{O}$ notation hides polylog factors of $n$. 
All these results were for gathering, but gathering with detection of two robots in this setting was studied by Elouasbi and Pelc~\cite{EP17}. They showed that when robots are allowed to communicate via the beeping model, gathering with detection can be solved deterministically in $\Tilde{O}(n^5 \log \ell)$ rounds. Bouchard et al.~\cite{BDP20} considered the even weaker model where robots may only detect the number of other robots co-located with them on a node and deterministically solved gathering with detection in time polynomial in $N$ and $\log \ell$, where $N$ is an upper bound on $n$ and is known to all robots, and in time exponential in $n$ when an upper bound on $n$ is not known to all robots. An important note is that in the aforementioned papers except~\cite{BDP20}, the robots did not know the value of $n$ initially while in the current paper, we make this assumption. However, in contrast to the papers on gathering~\cite{DFKP06,KM08,TZ14}, we solve the harder problem of gathering with detection. And in contrast to the paper on gathering with detection in the beeping model~\cite{EP17}, in the current paper, the robots do not know the value of $k$ initially (it can be argued that since~\cite{EP17} studies gathering specifically for two robots, this knowledge is implicit in their protocol design). And in contrast to~\cite{BDP20}, our algorithm is faster when there are more robots as theirs uses the algorithm from~\cite{TZ14}.

Regarding lower bounds for deterministic solutions to the problem of gathering with detection, we are not aware of any non-trivial lower bound. A trivial lower bound is $\Omega(n)$ (consider two (groups of) robots at either end of a line, it takes at least $n/2$ rounds for them to meet).

There has also been work on alternative settings, for example when the robots may experience crash faults or Byzantine faults, and alternative performance metrics, for example when solutions should be optimized both for time (as studied in the current paper) and for cost, i.e., the total number of edge traversals by all robots. These and other settings and metrics are mentioned in the survey by Pelc~\cite{P19}. 
In addition to work on arbitrary graphs, another approach to studying the problem has been to restrict the input graph to a specific class, such as a ring. A good resource for solutions on the ring (and a few other graph classes) is the book by Kranakis et al.~\cite{KKM10}.

This approach of leveraging the availability of multiple robots to achieve a faster solution has been used in other problems in the mobile robots on a graph setting, for example in exploration~\cite{Dereniowski:2015,DMNSS17}.

\section{Gathering on Arbitrary Graphs}\label{sec:algorithm}
In this section, we present our main algorithm for gathering with detection. We first present a gathering algorithm using universal exploration sequence (UXS) in Section~\ref{sec:gathering-via-uxs}, which works for any number of robots and detects the gathering. However, the UXS based algorithm takes long time, $\tilde{O}(n^5)$ rounds. We use this algorithm as a subroutine later in the main algorithm.  

Then we present a faster gathering algorithm with detection. For the faster algorithm, 
we first show how to solve the problem on a slightly limited set of input configurations of robots on nodes, and then show how to extend this algorithm to work for all input configurations. We interchangeably use the terms `distribution' and `configuration' of the robots.

\subsection{Gathering with Detection using Universal Exploration Sequence (UXS)}\label{sec:gathering-via-uxs}
In this section we show that gathering can be achieved with detection (i.e, robots can detect that gathering is completed and subsequently they can terminate) for any number of robots on an arbitrary $n$-node graph, where robots have the knowledge of $n$ but do not know $k$. For this we use the following result on universal exploration sequence (UXS), which implicitly follows from \cite{TZ14}. 

{\em A single robot with the knowledge of $n$ in an arbitrary anonymous graph can compute a bounded size UXS, and using the UXS it can explore the graph in time $\tilde{O}(n^5)$. As the length of the sequence is bounded, the robot terminates.} Let $T$ be an upper bound of the exploration time, i.e.,  $ T= \tilde{O}(n^5)$, for some large constant, and let $M$ be an upper bound on the memory required to implement the UXS.

Using the above result and the ID bits of robots, we can gather any number of robots with detection deterministically. Robots may have different lengths of ID bit-string. Further, robots only know $n$ and no other parameters. The gathering algorithm works as follows.  

\medskip
Initially, the robot(s) starting from same node form a group. If there is only one robot at some node,  it forms a singleton group. Every robot in a group follows the largest ID robot during graph exploration. Every robot that is not following any robot (i.e., the largest ID robot in a group or a singleton robot), reads their ID bits (one by one) from the least significant bit to the most significant bit and does the following in parallel. 

\begin{enumerate}[(i)]
\item If \textbf{the bit is $1$}, the robot explores the graph using UXS for $T$ rounds, and then waits at the node where it finishes the exploration for the next $T$ rounds. Thus, it takes $2T$ rounds. Recall that the other robots in the group (if any) simply follow it. If two or more groups meet during any of those $2T$ rounds, the respective groups merge. The robots start following the largest ID robot of the merged group thereafter.

\item \textbf{If the bit is $0$}, the robot waits at the current node for the first $T$ rounds, and for the next $T$ rounds, it explores the graph using UXS. It takes $2T$ rounds. Similarly, if two or more groups meet during any of these $2T$ rounds, the respective groups merge. The robots start following the largest ID robot of the merged group thereafter.

\item If the (largest ID) robot finishes scanning all its ID bits, it waits for another $2T$ rounds.  If it does not meet any group during this time, it terminates. Else, the respective groups merge and the robot starts following the largest ID robot of the merged group thereafter. 

\end{enumerate}

The high level idea behind such type of exploration and waiting for each bit is to make sure that robots can terminate. Consider a robot that is waiting for $2T$ rounds. In this period if no other robot meets this waiting robot, that means there is no robot that is still working on its bits, else that robot must meet this waiting robot irrespective of that robot working on its $0$ or $1$ bit.      

We show that the above procedure correctly gathers any number of robots in $O(T\log L)$ rounds, where $L$ is the largest ID of the robots. In addition, the robots detect that the gathering is completed and subsequently terminate. The correctness of the algorithm follows from the following lemmas.

\begin{lemma}\label{lem:2T-terminate1}
In step (iii), a robot say, $r$, which is not following any robot, is waiting for $2T$ rounds. The robot $r$ meets a group of exploring robots while waiting if and only if the length of the ID of the largest ID robot, say $r'$, in that group is larger than the length of the ID of $r$. 
\end{lemma}
\begin{proof}
First we show the `only if' part.

Let us assume that the length of ID of $r'$ is equal to the length of ID of $r$. Then both $r$ and $r'$ finish scanning their ID bits at the same time and both of them must wait at the same time. That is, while $r$ is waiting, $r'$ also must be waiting.

Now let us assume that the length of ID of $r'$ is less than the length of ID of $r$. As $r'$ is the largest ID robot in its group, it must have finished scanning its ID before $r$ does. So it is not possible that $r'$ comes and meets $r$ while $r$ is waiting.

Let us now argue the `if' part. Let $r'$ be the robot whose ID length is more than the ID length of $r$ and $r$ is waiting for $2T$ rounds. That is, $r'$ still has some bits left (at least one) to scan. Therefore, $r'$ must do an exploration of the graph, irrespective of the fact that its corresponding bit is $0$ or $1$. So, $r'$ meets $r$ during its waiting time of $2T$ rounds.   
\end{proof}

\begin{lemma}\label{lem:2T-terminate2}
In step (iii), let $r$ be a robot that is not following any other robot. If $r$ waits for $2T$ rounds and no robots meet $r$ at this time, the gathering is completed. 
\end{lemma}
\begin{proof}
If no robot meets $r$ during its wait, it follows from Lemma~\ref{lem:2T-terminate1} that there is no robot with larger ID present in the graph which is still working on some of its bits. This implies that all the other robots which are not following any robots, either have same length ID as $r$, or have lesser length ID than that of $r$. 

First consider the case where a robot $r'$ which is not following any robot and has same length ID as of $r$. This implies that both $r$ and $r'$ finished working on their bits at the same time. Since ID of $r$ and $r'$ must differ at some bit, they must have gathered during the working on that particular bit.

Now consider the case where the robot $r'$ which is not following any robot and have lesser ID length than that of $r$. Since $r'$ is not following anybody, it must finished scanning all its bits before $r$ did. So when $r'$ waited for $2T$ rounds, it must have met by $r$ (see Lemma~\ref{lem:2T-terminate1}).  

This completes the proof.
\end{proof}

\begin{lemma}\label{lem:2T-terminate3}
In step (iii), let $r$ be a robot that is not following any other robot. If $r$ terminates after $2T$ rounds, the termination is correct.
\end{lemma}
\begin{proof}
This is straightforward from Lemma \ref{lem:2T-terminate2}, since $r$ terminates only after gathering is completed.
\end{proof}
\begin{lemma}\label{lem:2T-terminate4}
Let $r$ be a robot that is following some other robot. The robot $r$ always terminates correctly.
\end{lemma}
\begin{proof}
Lemma~\ref{lem:2T-terminate3} says that each robot that is not following any robots terminates correctly. Now $r$ is following some robot, say $r'$, which means that  $r'$ is not following any other robot. So $r$ terminates only when $r'$ terminates. Since $r'$ terminates correctly, $r$ terminates correctly. 
\end{proof}

\begin{lemma}\label{lem:2T-terminate5}
The algorithm runs for $O(T\log L)$ rounds.
\end{lemma}
\begin{proof}
Let $r$ be the largest ID robot whose ID length is $O(\log L)$ as we assume $L$ is the ID. According to our algorithm, $r$ must finish working on all of its bits. For each bit, it takes $2T$ rounds. After finishing all the bits,  $r$ waits for another $2T$ rounds and then terminates. This makes the time complexity $O(T\log L)$ rounds. 
\end{proof}

Recall that $T = \tilde{O}(n^5)$. Thus, the time complexity of this UXS based algorithm becomes $\tilde{O}(n^5)$ rounds, since we assume the robots' ID lies in the range $[1, n^b]$ for some constant $b$.  Lemmas~\ref{lem:2T-terminate3}--\ref{lem:2T-terminate5} imply the following result.

\begin{theorem}\label{thm:main-via-uxs}
Consider a connected, undirected, anonymous graph of $n$ nodes and some robots (any number) are distributed over the nodes arbitrarily. Robots only know the value of $n$. Then there is a deterministic algorithm that gathers all the robots at some node with detection in $\tilde{O}(n^5)$ rounds, where the IDs of the robots lie in the range $[1, n^b]$ for some constant $b$. 
The algorithm requires that robots know only the value of $n$. Each robot requires $O(M + \log n)$ bits of memory, where $M$ is the memory required to implement the UXS.  
\end{theorem}

\medskip 
In the following two sections, we present a faster gathering algorithm. We first consider a setting where the initial distribution of the robots is {\em undispersed}, defined to be a distribution of the robots over the nodes where there is at least one node with two or more robots. After describing how to solve the problem of gathering in this setting in Section~\ref{sec:un-dispersed-config}, we then show in Section~\ref{sec:dispersed-config} how to extend our results to not only handle undispersed initial distributions of robots but also {\em dispersed} distributions, defined to be the situation where each node initially holds at most one robot. Notice that this distinction is useful in situations where the number of robots $k$ is at most the number of nodes $n$, i.e. $k \leq n$. When $k > n$, it is easy to see that there will always exist one node with more than one robot on it (also called as {\em Pigeonhole principle}) and as such any initial setup where $k>n$ is a undispersed distribution. Since the union of undispersed and dispersed initial configurations is the set of all input configurations, the algorithm presented in Section~\ref{sec:dispersed-config} works for all input configurations. We show that the algorithm from dispersed distribution gathers the robots {\em faster} if either at least two robots located at two nearby nodes (neighbors or a few hop distance away) or there are `many' robots. This validates the intuitive fact that the dispersed configuration is the worst configuration for the gathering problem. 

\subsection{Gathering with Detection from Undispersed Configuration}\label{sec:un-dispersed-config}
In the undispersed configuration, initially there is at least one node with multiple robots. A robot can have one of the following three states: (i) $finder$ (ii) $helper$ (iii) $waiter$. If a robot is not alone in the initial configuration and its ID is the minimum among the co-located robots, it sets its state as $finder$. If a robot is not alone in the initial configuration and its ID is not the minimum among the co-located robots, it sets its state as $helper$. If a robot is alone in a node in the initial configuration, it sets its state as $waiter$. Each robot maintains a variable $group_{id}$. Initially the $group_{id}$ of each $finder$ robot is its own ID whereas each $helper$ robot stores the ID of the co-located $finder$ robot in its $group_{id}$. Each $waiter$ robot puts -1 in its $group_{id}$. The algorithm runs in phases. The first phase is devoted to \textit{map finding} and the second phase is for \textit{gathering}. Each robot can detect when the gathering is complete. Let us first provide a high level idea before explaining the algorithm. 

As each $finder$ robot has a company of $helper$ robots, the $finder$ robot finds a map of the graph using an existing algorithm in Phase 1. Then in Phase 2, each $finder$ robot explores the graph following the map it posses. The $helper$ and $waiter$ robots remain at their initial location until some $finder$ robot arrives and pick them up.  During the exploration, the $finder$ robots collect the $helper$ and $waiter$ robots. Though, at this point, it seems the robots may end up in multiple groups, later we describe how one particular $finder$ robot gathers all the other robots to its initial position after just one run of graph exploration. Below we present the full algorithm.

\medskip
\noindent Algorithm {\sc \textbf{Undispersed-Gathering:}}

\medskip
\noindent\textbf{Phase 1 (map finding):} The robots with state  $finder$ and $helper$ take part in this phase. The $waiter$ robots remain at their position. Each robot with state $finder$ works as an agent and the remaining co-located robots (with state $helper$) work as a movable token. Collaboratively they run the  exploration algorithm with movable token (for map finding) presented in \cite{DPP14}. It is possible for each  $finder$ to compute an isomorphic map of the underlying graph in $O(n^3)$ rounds. Since each $finder$ robot needs to store the map of the graph, the memory requirement becomes $O(m\log n)$ where $m$ is the number of edges of the graph. After constructing the map, each $finder$ meets the $helper$ robots which were working as its token during this phase, and provides the value of $n$ to the helper robots. Note that, even if multiple finder-helper(s) combination  perform this algorithm in parallel, still each finder can keep track of its helpers (and vice-versa), as the $group_{id}$ of its helpers must be equal to its own ID.

Let the number of rounds in Phase~1 be denoted as $R_1 = O(n^3)$. After Phase 1 completes, each $finder$ robot knows an isomorphic copy of the map of the graph. The $finder$ and $helper$ robots wait until $R_1$ rounds are over. Then, from round $R_1+1$, they execute Phase 2.

\medskip
\noindent\textbf{Phase 2 (gathering):} All the robots take part in this phase. Below we provide the tasks of the robots in a round with different states.
\begin{itemize}
    \item \textbf{Algorithm for $finder$ robots:} each $finder$ robot computes a spanning tree from the graph (map) it posses and does exploration along the edges of the spanning tree. Let a $finder$ robot $f$ be at some node at the start of some round. It  communicates with the co-located robots to know each other's states as well as $group_{id}$.
        \begin{itemize}
            \item If there are only $waiter$ robots present in that node, the $finder$ provides the value of $n$ and moves to another node to continue exploration according to its spanning tree unless it is back to the node from which it started its exploration. 
    
            \item If there are other $finder$ and/or $helper$ robot(s) present in that node, $f$ provides moves to another node to continue exploration according to its spanning tree unless it is back to the node from which it started its exploration if its $group_{id}$ is minimum among the co-located robot's $group_{id}$; else if there is a $finder$ robot with the minimum $group_{id}$, $f$ changes its state to $helper$, updates its $group_{id}$ to that $finder$ robot's $group_{id}$, and starts following that robot; else if there is a $helper$ robot with the minimum $group_{id}$, $f$ changes its state to $helper$, updates its $group_{id}$ to that $helper$ robot's $group_{id}$ and stays at this node.
    
            \item If there are only $helper$ robots present in that node, $f$ moves to another node to continue exploration according to its spanning tree unless it is back to the node from which it started its exploration, if there are no $helper$ robots with lesser $group_{id}$ than the $group_{id}$ of $f$; else $f$ changes its state to $helper$, updates its $group_{id}$ to the minimum $group_{id}$ among the $helper$ robot's $group_{id}$ and stays there. 
        \end{itemize}
        
    \item \textbf{Algorithm for $helper$ robots:} A $helper$ robot stays at the node if no $finder$ robot arrives in that node in the previous round. If one or more $finder$ robots arrive, if there is at least one $finder$ robot whose $group_{id}$ is less than the $helper$ robot's  $group_{id}$, the $helper$ robot changes its $group_{id}$ to the $group_{id}$ of the $finder$ robot whose $group_{id}$ is minimum among all the co-located $finder$ robots and starts following it. 
    
    \item \textbf{Algorithm for $waiter$ robots:} A $waiter$ robot stays at the node if no $finder$ robot arrives in the previous round. Else, if one or more $finder$ robots arrive, the $waiter$ starts following the minimum $group_{id}$ $finder$ robot.  It also changes its state to $helper$ and update its $group_{id}$ to the $group_{id}$ of the minimum $group_{id}$ $finder$ robot.

    \item \textbf{Termination:} Each robot keeps a counter of number of rounds since the start of the algorithm. When the counter equals $R_1+2n$, each robot terminates.
\end{itemize}

\bigskip
\noindent\textbf{Correctness:} \\
Correctness of Phase 1 follows from \cite{DPP14}. Here we study the correctness of Phase 2. Let $f$ be that $finder$ robot whose $group_{id}$, say, $l$ is the minimum among all the $finder$ robots. Let $S$ be the set of $helper$ robots with $group_{id}$ equal to $l$ at the beginning of Phase 2. It is easy to observe that $S$ is non-empty, else $f$ would not have been a $finder$ robot. Let $v$ be the node where $f$ and the $helper$ robots of $S$ belongs to at the beginning of Phase 2. We claim the following.

\begin{lemma}\label{lem:corr}
By the time the finder robot $f$ with smallest $group_{id}$ completes its graph exploration and comes back to the node $v$ from where it started Phase 2, all other robots are gathered at $v$. 
\end{lemma}
\proof
In the very first round of Phase 2, $f$ definitely finds its $group_{id}$ $l$ to be the minimum among the co-located robot's  and starts graph exploration. Note that the $group_{id}$ of the robots of $S$ is also $l$ but this does not restrict $f$ to consider its $group_{id}$ to be the minimum among co-located robots. It is easy to see that, according to our algorithm, $f$ never changes its state and comes back to $v$ after completing the graph traversal as its $group_{id}$ always remains the minimum among all.  

Consider any $helper$ robot $h$ with $group_{id}$ larger than $l$. The robot $h$ stays at its position till some $finder$ robot with lesser $group_{id}$ arrives. If the first $finder$ that arrives at the node where $h$ resides happens to be $f$, then $h$ reaches $v$ along with $f$ as $h$ continue to follow $f$. If some other $finder$ robot $f'$ with lesser $group_{id}$ meets $h$ first, then $h$ starts to follow $f'$. Definitely $f'$ did not explore node $v$ yet else it would have been stuck there. This continues unless $f'$ meets another finder $f''$ and becomes a $helper$. Again, it is definite that $f''$ did not visit $v$ yet. This continues till $h$ reaches at $v$. The $helper$ $h$ does not reach $v$ at some round implies that the $finder$ robot whom $h$ is currently following, has not yet finished its exploration. Since our algorithm proceeds in synchronous rounds,  by this round $f$ is also not done with its exploration then. This shows each $helper$ robot eventually reaches $v$ by the time $f$ completes exploration and returns back to $v$.

Consider any $waiter$ robot $w$. According to our algorithm, whenever $w$ meets some $finder$ robot, it changes its state to $helper$, updates $group_{id}$ and starts following the $finder$ robot.  As we have already shown that any $helper$ robot eventually reaches $v$, this shows each $waiter$ robot eventually reaches $v$ by the time $f$ completes exploration and returns back to $v$.

Consider any $finder$ robot $f'$ other than $f$. According to our algorithm, the following three cases are possible.
\begin{itemize}
    \item $f'$ reaches $v$ while exploring the graph as a $finder$ robot. In this case $f'$ changes its $group_{id}$ to $l$, changes its state to $helper$ and stays at $v$. This is because $l$ is the minimum possible $group_{id}$ among all available $group_{id}$s.
    
    \item $f'$ does not reach $v$ as a $finder$. This happens in the following two cases. 
        \begin{itemize} 
            \item During exploration $f'$ meets some $finder$ robot $f''$ with lesser $group_{id}$ and started follow $f''$ after changing its state to $helper$. As we have already shown that any $helper$ robot eventually reaches $v$, we are done with this case.
            
            \item During exploration $f'$ meets some $helper$ robot with lesser $group_{id}$, changes its state to $helper$ and stay at that node. As we have already shown that any $helper$ robot eventually reaches $v$, we are done with this case.
        \end{itemize}
\end{itemize}

It is straightforward to observe that this phase gathers robots in $2n$ rounds as the minimum ID finder robot requires exactly $2n$ rounds to explore all the nodes of the graph along some spanning tree that the $finder$ computes using the copy of the map of the graph it posses.  Also, all the $finder$ and $helper$ robots know $n$ by the end of phase 1 and each $waiter$ robot receives the value of $n$ from some $finder$ robot in phase 2 since each waiter robot meets at least one $finder$ robot in this phase. Hence it is possible for each robot to terminate once the counter equals $R_1+2n$.
\endproof

Now we have the following main theorem of this section. 
\begin{theorem}\label{thm:main-undispersed}
Given an $n$-node anonymous graph (undirected and connected) and $k$ robots are distributed over the nodes arbitrarily such at least one node holds more than one robot, then the  deterministic algorithm {\sc Undispersed-Gathering} gathers all the robots to a single node in $O(n^3)$ time using $O(m\log n)$ memory per robot, where $m$ is the number of edges in the graph. Also, each robot can detect when the gathering is completed. 
\end{theorem}
\proof
The correctness of Phase 1 of our algorithm follows from \cite{DPP14}. The correctness of Phase 2 of our algorithm that gathers all the robots is already provided in Lemma \ref{lem:corr}. The time complexity of Phase 1 is $O(n^3)$ and follows from \cite{DPP14}. As any robot except the minimum ID robot can be a $finder$ robot, the memory requirement due to storing a map of the graph is $O(m \log n)$ bits per robot. The time complexity of Phase 2 depends on the time required by the $finder$ robot with minimum $group_{id}$ to complete the graph exploration and return to its initial position at the beginning of Phase 2. As the $finder$ explores according to a spanning tree of the graph, it needs exactly $2n$ rounds. Hence the overall time complexity becomes $O(n^3)$.

\endproof

Let the running time of {\sc Undispersed-Gathering} be upper bounded by $R$, where $R = R_1+ 2n \in O(n^3)$. We use this later in presenting the main algorithm which works for any initial distribution. 

\subsection{Gathering from Dispersed Configuration}\label{sec:dispersed-config}
In this section, we first present an algorithm for gathering the robots from the dispersed configuration. In particular, given the $k \leq n$ robots initially positioned at a dispersed configuration, we present an efficient approach to reach an undispersed configuration. Then applying the above algorithm in Section~\ref{sec:un-dispersed-config}, it solves the gathering problem. To make it lucid, let us first consider the following simpler case. 

Suppose, in the dispersed configuration, two robots are positioned in two neighboring nodes (i.e., $1$-hop away). We show that it is possible to convert this configuration to an undispersed configuration in $O(n \log n)$ time. For this, we use the ID bit-string of the robots to assemble them to a single node, and hence reaching to an undispersed configuration. Let the two robots $r_u$ and $r_v$ be positioned on two neighboring nodes $u$ and $v$ respectively. Since the IDs of the robots are distinct, the ID bit-string of $r_u$ and $r_v$ are different (and possibly different lengths). The robots $r_u$ and $r_v$ run the following procedure to assemble.     

\medskip
\noindent {\sc \textbf{$1$-Hop-Meeting:}} The approach runs in cycles and each cycle consists of $2(n-1)$ rounds. A robot performs the following by looking at the bits of its ID from `right' to `left' (i.e., in the reversed order). If the bit is $0$, the robot doesn't move and stays at its node for the $2(n-1)$ rounds. If the bit is $1$, the robot visits all the neighbors one by one following the port numbers starting from $1$. Since visiting a neighbor takes $2$ rounds (go to the neighbor and come back), the robot must complete visiting all the neighbors in $2(n-1)$ rounds.\footnote{If  the maximum degree of the graph $\Delta$ is known to the robots, then visiting neighbors may take $2\Delta$ rounds. However, assuming $\Delta = n-1$ will not affect the asymptotic bound in our main result for some cases.} If the degree of the node is less than $n-1$, the robot waits at its node for the remaining rounds of the cycle. Further, since the ID lengths are  different, a robot may finish scanning all its ID bits earlier than the others. Once a robot finishes scanning its ID bits, it waits for the procedure to be ended--- which is $a n \log n$ rounds, for some large constant $a$\footnote{Technically, the value of $a$ should be larger than $b$ in the ID range $[1, n^b]$. If such a value of $a$ is not known, $a$ can be taken as $\log\log n$. In this case, the $\log n$ factor in all the running time bounds below will be replaced by $(\log n \log\log n)$. For simplicity, we assume $a$ to be a large constant.}. The robots wait in both the cases to synchronize the start/end time of the cycles. Therefore, it takes one cycle time for one bit. Since we assume the IDs are of length $O(\log n)$ bits, the procedure stops in $O(n\log n)$ rounds.

The above procedure guarantees the meeting of two neighboring robots. Since the ID bit-strings of the robots are different, there is at least one index in the string where their bits are different, i.e., one is $0$ and another is $1$. For the first such an index, one of them doesn't move and the other visits all the neighbors. Thus, they meet and assemble there. Notice that the procedure doesn't guarantee assembling any pair of neighboring robots, but ensures assembling at least one pair and hence reaching to an undispersed configuration. Therefore, we get the following lemma. 

\begin{lemma}\label{lem:meeting-two-neighboring-robots}
Suppose two robots are located on the two neighboring nodes in a dispersed configuration. Then it takes $O(n\log n)$ rounds to reach an undispersed configuration.  
\end{lemma}

The {\sc $1$-Hop-Meeting} procedure can be extended to the meeting of two robots which are at $i$-hop distance away from each other. Let us call the procedure as {\sc $i$-Hop-Meeting}. 

\medskip
\noindent {\sc \textbf{$i$-Hop-Meeting:}} In the same way, the approach runs in cycles where each cycle consists of $T(i) = \sum_{j=1}^i 2(n-1)^j$ rounds. A robot performs the following by looking at the bits of its ID from `right' to `left' (i.e., in the reversed order). If the bit is $0$, the robot doesn't move and stays at its node for the $T(i)$ rounds. If the bit is $1$, the robot visits all the nodes lie within the $i$-hop distance from it using DFS traversal following the port numbers. In the same way, if the degree of some nodes is less than $n-1$ the robot waits at its node for the remaining rounds of the cycle. A robot also waits for $a n \log n$ rounds at its node if finishes scanning all its ID bits, where $a$ is some large constant.  

 The procedure takes $(\sum_{j=1}^i 2(n-1)^j\log n)$ rounds, which is $O(n^i\log n)$ rounds. Thus, we get the following general result. 

\begin{lemma}\label{lem:meeting-i-hop-robots}
Suppose, in the dispersed configuration, two robots are positioned on the two nodes which are at a $i$-hop distance away from each other. Then it takes $T(i) = \sum_{j=1}^i 2(n-1)^j\log n =O(n^i\log n)$ rounds to reach an undispersed configuration, where $i\leq D$, the diameter of the graph.  
\end{lemma}

Therefore, it follows from the above lemma that when $i > 5$, the procedure {\sc $i$-Hop-Meeting} itself takes a longer time than the existing gathering result--- $\tilde{O} (n^5 \log n)$ rounds. For $i\leq 5$, our algorithm outperforms the best existing algorithms.

\medskip
Therefore, the complete algorithm, for any number of robots and for any initial distribution of the robots, works as follows. 

\medskip
\noindent Algorithm {\sc \textbf{Faster-Gathering:}}

\medskip
\noindent \textbf{Step~1:} Every robots run the {\sc Undispersed-Gathering} algorithm. If, indeed, the initial distribution of the robots is undispersed configuration, then gathering is successful and the robots terminate if it is not alone (see the Lemma~\ref{lem:main-termination}). Otherwise, there won't be any movement of the robots and all the robots stay at their original position (node) for $R$ rounds. 

\medskip
\noindent (\textbf{Step~2} to \textbf{Step~6}): For $i = 2, \dots, 6$, run the following steps.

\noindent \textbf{Step~i:} After $(i-1)R$ rounds, if gathering is not achieved (i.e., the initial distribution is dispersed), every robots first run the {\sc $(i-1)$-Hop-Meeting} procedure and then run the {\sc Undispersed-Gathering} algorithm. If  at least two robots located at neighboring nodes, then {\sc $1$-Hop-Meeting} converts the initial distribution to an undispersed configuration. Then by running {\sc Undispersed-Gathering} solves the gathering problem.  Each robot terminates if it is not alone (see Lemma~\ref{lem:main-termination}). Otherwise, there won't be any movement of the robots and all the robots stay at their position for $R$ rounds. 

\medskip
\noindent \textbf{Step~7:} After $6R$ rounds, if gathering is not yet achieved, run the gathering algorithm using UXS presented in Section~\ref{sec:gathering-via-uxs}. Since this algorithm works for any number of robots and for arbitrary initial distribution of the robots, gathering must be achieved with the detection.

\medskip
We now show that every robots detect that gathering is achieved in the end of every steps of the algorithm.  

\begin{lemma}\label{lem:main-termination}
In the end of any of the first 6 steps of the {\sc Faster-Gathering} algorithm, if a robot is alone at some node, then each robot is alone.
\end{lemma}
\proof
Each of the first 6 steps call {\sc Undispersed-Gathering} algorithm at the end. There can be the following two cases. 
\begin{itemize}
\item At the time when {\sc Undispersed-Gathering} algorithm starts in any particular phase and the robot configuration is dispersed. In this case, all the robots are $waiter$ robots and do nothing. Accordingly at the end of this step, each robot is alone as the dispersed configuration never changes.
\item At the time when {\sc Undispersed-Gathering} algorithm starts in any particular phase and the robot configuration is undispersed. Then by Theorem \ref{thm:main-undispersed}, no robot will be alone at the end of the step. \endproof
\end{itemize}

Thus, it follows from the above lemma that a robot can detect gathering via alone or not. Finally, the Step~7 also guarantees gathering with detection.   

The {\sc Faster-Gathering} algorithm correctly gathers the robots at a single node, detect the completion of the gathering and terminates. Each step can be synchronized easily using the time bound of {\sc Undispersed-Gathering} and {\sc $i$-Hop-Meeting}. Therefore, the time complexity of the algorithm can be written as $O(\min\{R + T(i), \tilde{O} (n^5)\}$ rounds for $i = 0, 1, 2, \ldots, 5$, where $T(0) = 0$ indicates an undispersed distribution, $T(i) =O(n^i\log n)$, for $i= 1, 2, \ldots, 5$, see Lemma~\ref{lem:meeting-i-hop-robots} and $\tilde{O} (n^5)$ is the time complexity of the gathering algorithm using UXS, see Theorem~\ref{thm:main-via-uxs}. In the following, we present the main result. 

\begin{theorem}[Faster Gathering]\label{thm:main}
Given an $n$-node, $m$-edge anonymous graph (undirected and connected) and $k$ robots are distributed over the nodes arbitrarily, then there is a deterministic algorithm which gathers all the robots to a single node with detection in time:
\begin{enumerate}[(i)]
    \item $O(n^3)$ rounds, if the initial distribution of the robots is either {\em undispersed}, or {\em dispersed} with at least two robots positioned at a distance $2$ from each other.
    \item $O(n^i \log n)$ rounds, if the initial distribution of the robots is {\em dispersed} with at least two robots positioned at a distance $i$ from each other, for $i= 3, 4$ and $5$.
    \item Otherwise, $\tilde{O}(n^5)$ rounds.   
\end{enumerate}

The algorithm requires that robots know only the value of $n$. Each robot requires $O(M + m\log n)$ bits of memory, where $M$ is the memory required to implement the UXS.  
\end{theorem}

\begin{remark}\label{rem:given-hop-info}
We note that if the hop distance information in the initial configuration is given, then the {\sc Faster-Gathering} algorithm finishes faster by directly running the particular step of the algorithm. 
\end{remark}

\begin{remark}\label{rem:onDelta}
Let us further remark that if $\Delta$, the maximum degree of the graph, is known to the robots, the time complexity of {\sc Faster-Gathering} becomes $O(\min\{R + \Delta^i\log n, \, \tilde{O} (n^5)\})$ rounds for $i \leq 5$. It follows from the procedure {\sc $i$-Hop-Meeting}, in which, each cycle consists of  $T(i) = \sum_{j=1}^i 2\Delta^j$ rounds to visit all the neighbors of a node. 
\end{remark}
Now, we show a crucial result on the fact that when there are many robots, one cannot place all of them far from each other in the dispersed configuration. In other words, if there are sufficiently many robots, at least two of them must be positioned nearby (say, at most $5$-hop away from each other). It follows from a more general result stated below.

\begin{lemma}\label{lem:hop-distance-bound}
Suppose $\left\lfloor\frac{n}{c}\right\rfloor+1$ robots are distributed arbitrarily over the nodes on a $n$-node graph, for any constant $c$. Then there exists at least two robots which are at most $2c-2$ hop distance away from each other.
\end{lemma}

\proof
The number of robots is $\lfloor\frac{n}{c}\rfloor+1\geq 2$, since $c\leq n$. By contradiction, assume that any pair of robots have distance at least $2c-1$. We show a contradiction. Consider a robot, say, $r$, located at some node $v$ and let $v'$ be the node that contains another robot. Since the graph is connected, $v$ must be connected to $v'$. Let $v_1$, $v_2$, $\cdots$, $v_{c-1}$ be the first $c-1$ nodes on a shortest path from $v$ to $v'$. So the distances of $v_1$, $v_2$, $\ldots$, $v_{c-1}$ from $v$ are $1, 2, \ldots, c-1$ respectively. Then by the assumption, no robot $r'$ can be positioned at a distance less than $c$ from any of the nodes $v_1$, $v_2$, $\ldots$, $v_{c-1}$; otherwise, the distance between $r$ and $r'$ becomes less than $2c-1$. This implies that for any robot $r$, there must exists a path of at least $c-1$ nodes which are free, i.e., not holding a robot (since the graph is connected). 
It is easy to observe that, $v_i$ can not lie on a free path of any other robot, say $r'$. If it is the case, then $r'$ is a robot that lies within $c$ distance from $v_i$ and it is a contradiction since $v_i$ is a part of a free path of $r$.

Since $r$ is any arbitrary robot, the above is true for all the robots. This implies there must exist at least one set of such $c-1$ designated nodes for each robot positioned at some node.

As we started with $\lfloor\frac{n}{c}\rfloor+1$ many robots, we have $(c-1)\lfloor\frac{n}{c}\rfloor+(c-1)$ distinct nodes corresponding to those robots. Further, $\lfloor\frac{n}{c}\rfloor+1$ robots positioned at $\lfloor\frac{n}{c}\rfloor+1$ nodes. Hence, the total number of nodes in the graph required to place $\lfloor\frac{n}{c}\rfloor+1$ robots with the assumed condition is: $(c-1)\lfloor\frac{n}{c}\rfloor+(c-1) + \lfloor\frac{n}{c}\rfloor+1 = c\lfloor\frac{n}{c}\rfloor+c \geq n+1$. This is a contradiction as the number of nodes in the graph is $n$. So our assumption that any pair of robots has distance at least $2c-1$ is wrong. Hence the proof.
\endproof
Then, from the above results, namely,  Theorem~\ref{thm:main}, Lemma~\ref{lem:meeting-i-hop-robots},  Lemma~\ref{lem:main-termination}  and Lemma~\ref{lem:hop-distance-bound}, we get the following results on gathering with detection. 

\begin{theorem}[Gathering with Detection]\label{thm:gathering-with-detection}
Given an $n$-node anonymous graph (undirected and connected) and $k$ robots are distributed over the nodes arbitrarily, then the robots can be gathered to a single node deterministically and every robot detects the gathering in time: 
\begin{enumerate}[(i)]
    \item $O(n^3)$ rounds, if $k \geq \left\lfloor\frac{n}{2}\right\rfloor+1$.
    
    \item $O(n^4 \log n)$ rounds, if  $\left\lfloor\frac{n}{3}\right\rfloor+1 \leq k < \left\lfloor\frac{n}{2}\right\rfloor+1$.
    
    \item $\tilde{O}(n^5)$ rounds, if  $ k < \left\lfloor\frac{n}{3}\right\rfloor+1$. 
    
\end{enumerate}
The algorithm requires that robots know only the value of $n$. Each robot requires $O(M + m\log n)$ bits of memory, where $M$ is the memory required to implement the UXS.  
\end{theorem}
\proof
Let us show the bounds for all the cases one by one. 
\begin{enumerate}[(i)]
\item If there are $k \geq \left\lfloor\frac{n}{2}\right\rfloor+1$ robots, then the Lemma~\ref{lem:hop-distance-bound} ensures that there exists a pair of robots within $2$-hop distance. Then by Theorem~\ref{thm:main}, gathering is achieved in $O(n^3)$ rounds, and by Lemma~\ref{lem:main-termination}, every robot detects it.

\item Similarly, for $\left\lfloor\frac{n}{3}\right\rfloor+1 \leq k < \left\lfloor\frac{n}{2}\right\rfloor+1 $, the Lemma~\ref{lem:hop-distance-bound} ensures that there exists a pair of robots within $4$-hop distance. Then by Theorem~\ref{thm:main} and Lemma~\ref{lem:main-termination}, gathering with detection is achieved in $O(n^4\log n)$ rounds. 

\item For $k < \left\lfloor\frac{n}{3}\right\rfloor+1$,  gathering can be achieved in $\tilde{O}(n^5)$ rounds, from Theorem \ref{thm:main}. \endproof 
\end{enumerate} 

Therefore, it follows that the algorithm {\sc Faster-Gathering} solves the gathering with detection faster (than the existing results) if there are {\em many} robots in the system. To be specific, if there are at least $\left\lfloor\frac{n}{3}\right\rfloor+1$ robots, then our algorithm performs significantly faster than any existing algorithm as it requires no more than $O(n^4\log n)$ rounds to gather the robots (with detection). And if the number of robots is more than $\left\lfloor\frac{n}{2}\right\rfloor+1$, then our algorithm performs even better --- takes only $O(n^3)$ rounds. Note that, while having many robots is a sufficient condition for achieving faster gathering with detection, it is not a necessary condition. One can achieve the same with any number of $k$ robots ($k\geq 2$), if there exists a pair of robots located within a distance $5$ from each other on the graph in the initial configuration. This is evident from Theorem \ref{thm:main}.

\section{Conclusion and Future Work}
\label{sec:conclusion}
In this paper, we looked at the problem of gathering with detection of robots in a graph. We presented a deterministic algorithm that worked faster than all pre-existing deterministic algorithms for this problem and even for the easier problem of gathering without detection, subject to some conditions. 

However, in order for this algorithm to work, we assumed that all robots simultaneously woke up. An interesting future direction would be to see if we can leverage this approach of utilizing many robots for faster gathering, even if robots wake up at arbitrary times. Additionally, it would be interesting to see if we can still design fast algorithms if the communication capabilities of the robots are reduced, for example to just being able to beep or not. Finally, try to get a faster deterministic gathering algorithm for a small number of robots, say, whether it is possible a $o(n^5)$-time algorithm for a constant (or polylogarithmic) number of robots.   

An interesting line of future work is as follows. We do not restrict the size of messages exchanged between robots at a node. It would be interesting to consider the model where the size of messages is restricted and study the resulting effect on running time.

\bibliographystyle{plainurl}
\bibliography{references}


\end{document}